\documentclass{fundam}
\usepackage{url} 
\usepackage[ruled,lined]{algorithm2e}
\usepackage{graphicx}

\usepackage{booktabs}
\usepackage{makecell}
\begin{document}
\setcounter{page}{1001}
\issue{XXI~(2001)}

\title{Three-Party Integer Comparison and Applications}

\address{Corresponding author: Kewei Lv. State Key Laboratory of Information Security, Institute of Information Engineering, Chinese Academy of Sciences, Beijing, 100093, China}

\author{Jie Ma$^{a,b,c}$ \and Bin Qi$^{a,b,c}$ \and  Kewei Lv$^{a,b,c,}$\\
$^a$State Key Laboratory of Information Security\\
Institute of Information Engineering\\
Chinese Academy of Sciences\\
Beijing, 100093, China.\\
$^b$Data Assurance Communication Security Research Center\\
Chinese Academy of Sciences\\
Beijing, 100093, China.\\
$^c$School of Cyber Security\\
University of Chinese Academy of Sciences\\
Beijing, 100093, China.\\
majie@iie.ac.cn\\
qibin@iie.ac.cn\\
lvkewei@iie.ac.cn}\maketitle

\runninghead{J. Ma, B. Qi, K. Lv}{Three-Party Integer Comparison and Applications}

\begin{abstract}
  Secure integer comparison has been a popular research topic in cryptography, both for its simplicity to describe and for its applications. The aim is to enable two parties to compare their inputs without revealing the exact value of those inputs.

        In this paper, we highlight three-party integer comparison (TPIC), where a \emph{judge}, with no private input, wants to know the comparison result, while two \emph{competitors} hold secret integers to do privacy-preserving comparison. The judge actively obtains the result rather than passively waiting for it sent by a competitor.
        We give two TPIC constructions considering \emph{Mixed adversaries}, who have with different capabilities.
        One is secure against a semi-honest adversary with low computation and communication cost, while the other is secure against a malicious adversary.

        Basing on TPIC, we present multi-party comparisons through concrete applications, including a joint bidding scheme and a practical auction.
        Brief security proofs and analysis for the applications are presented. In comparison, our auction scheme is more efficient with lower cost, making it feasible in practice rather than a theoretical design.
        All the comparisons and application schemes run on top of blockchain requiring a constant number of rounds.
\end{abstract}

\begin{keywords}
Multi-Party Integer Comparison, Joint Bidding, Auction, Commitment, Constant-round, Blockchain
\end{keywords}

\section{Introduction}\label{sect-intro}
        Integer comparisons play an important role in the real life, a classic topic is the Millionaires problem \cite{Yao86}, bringing researches on two-party comparison. A large number of solutions have been proposed employing various methods like garbled circuits \cite{KSS09,AIKW15}, oblivious transfers \cite{CO15}, and homomorphic encryption \cite{Fisch,DGK,CEK}. One approach is to decompose integers into bitwise representation and evaluate them in a boolean circuit. The communications and computations are inefficient requiring many rounds of interaction. In contrast to bitwise comparisons, more efficient protocols were given in \cite{CEK,AbspoelBSV19,BST20,Eskeland20} to compare multiple bits simultaneously with a single ciphertext by a somewhat homomorphic encryption.

        The general case related to secure integer comparison includes multi-party ranking \cite{GMW87}, of which a concrete example is auction, such as English auction and uniform-price auction \cite{Bran}.
        The dramatic development of internet of things, 5G and electronic commerce, etc,  facilitates the research on online auctions \cite{GY,DGK,TianSMAA21,TejaM21,NiCXZZ21,SarencheSAA21}, which heavily rely on integer comparisons to determine one or more bidders with higher bids.
        Recently, the emergence of blockchain brings new idea to design comparison protocols as well auction schemes \cite{GY,Snc,KMSWP}. In the literature,
        Fischlin \cite{Fisch} proposed two-party bitwise comparison by boolean circuit, based on which, auction schemes were given considering semi-honest and malicious case.
        An oblivious third party A is required to provide a public/private key pair. All homomorphic computations in Fischlin's protocol are then performed under A's public key. Simulating A on the blockchain requires distributing the private key over multiple parties. As a result, one would need a secure and distributed computation of a Goldwasser-Micali key pair \cite{GM82}. Even for the case of RSA, this is complex and requires more rounds to do interactions \cite{BF1997}, making it impractical on blockchain.
        Blass and Kerschbaum \cite{BK2} proposed $Strain$, a protocol to implement auctions on top of blockchains that protects the bid privacy against fully-malicious parties. They improved the two-party comparison mechanism of Fischlin to achieve higher efficiency. Moreover, the authors mentioned that the proposed protocol leaks the order of bids to the public.
        Damg{\aa}rd et al. \cite{DGK} gave auction schemes by homomorphic encryption. They considered the very different scenario of comparing a secret value $m$ with a public integer $x$ for the fully malicious version while the semi-honest version addresses comparing secret inputs $m$ and $x$. They described their special case consisting of a server, an assisting server and a client with harsh condition that the client must be able to send his value and go offline, whereafter the other two parties should be able to do the computations together. The offline requirement is not very suitable for blockchain architecture pursuing freedom in joining and leaving the network at any time.
        The auction scheme \cite{MQL} used the two-party multi-bit comparison in \cite{CEK} considering a semi-honest adversary. Later, it was improved in \cite{MQL2}, which, however, required help of the auctioneer to forward messages in more rounds.

        We observe that two-party comparisons are pivots of most existing auctions. However, there are cases, where two-party protocols might not work well. Consider such a scenario:  a buyer (with balance $\alpha$) wants to buy a toy (of price $\beta$) from a seller via blockchain. The buyer needs to convince the miners that $\alpha\!\ge\!\beta$, while guaranteeing privacy of both $\alpha$ and $\beta$.
        By a two-party comparison protocol, the buyer obtains the comparison result and sends it to the miners with a proof, resulting in additional complexity. Can we design a protocol to enable the miners to actively obtain the result rather than passively receiving it from the buyer?

        \textit{Our contributions}. In this paper, we highlight the above question and address it by three-party integer comparison (TPIC) protocols.
        In TPIC, two \emph{competitors} hold their secret integers to perform comparisons, while a \emph{judge} with no private input learns the comparison result.
        Two \emph{constant-round} constructions are given against semi-honest and malicious adversary respectively. 
        One is built on \cite{CEK} for semi-honest competitors, and has low communication and computation cost (it can be improved by the method in \cite{BST20,Eskeland20} to reduce one round), while the other is built on the technique of Fischlin \cite{Fisch} for malicious competitors.
        In both, the judge is assumed to be semi-honest and non-colluding with competitors to maintain the comparison result.
        As applications of TPIC, we give schemes for joint bidding and online auction, which run on top of blockchain in a constant number of rounds.

        It is worth mentioning that the proposed joint bidding scheme runs in 4 blocks and is better than \cite{MQL}, where the number of rounds required is equal to the number of firms.
        Further, our auction has lower cost compared with related schemes \cite{DGK,Fisch} and is feasible in practice. It runs in a constant number of rounds (4 blocks) and does not require an intermediary as in \cite{MQL2}. It also avoids the issue mentioned in \cite{BK2}: the order of bids is revealed to all the participants. Indeed, the order of the private inputs is regarded as important information in some scenarios like joint bidding \cite{CL02}. In our auction, the judge can obtain the results on its own initiative, and is the only party to get the order.
        Complexity analysis is given, while there is no analysis in \cite{MQL2}.

        \textbf{Organization.} This paper is organized as follows. In section~\ref{sect-intro}, an introduction towards the research background as well as our contributions are presented. In section~\ref{sect-prelimi}, we give a description for preliminaries including the ideal functionality, the adversary model, as well as the commitment scheme. In section~\ref{sect-TPIC-scheme}, three-party integer comparison protocols are given. The applications to illustrate multi-party comparisons are given in section~\ref{sect-Appli}, consisting of joint bidding and auction schemes. The complexity analysis of the auction is shown in section~\ref{sect-complexity}. Finally, conclusions are given in section~\ref{sect-conclu}.

\section{Preliminaries}\label{sect-prelimi}
        In this section, we present the system model, ideal functionality and security definition, together with a commitment scheme used in the multi-party applications. 

        \textbf{Convention.} When we say someone ``broadcasts" something, it means that he publishes his message on the blockchain via a transaction.
        We denote by $\lambda$ as the security parameter involving in a cryptosystem.
  \subsection{System Model}\label{system-model}
        We give a general model for comparisons in this paper. The participants are divided into two parts: one is called \emph{judge} with no private input like the auctioneer in an auction, the other is called \emph{competitors} which is a set of participants holding private inputs like bidders. The judge wants to obtain the sorted order of the secret inputs, which the owners does not wants to reveal. Assume there are $N$ ($N\!\ge\!2$) competitors and one judge in the following. In TPIC, $N\!=\!2$.

        \noindent\textbf{Ideal Functionality.}
         Our integer comparison scheme emulates a trusted third party (TTP) that first receives private inputs from all competitors by an authenticated channel. TTP then compares all the inputs and sends all the comparison results $cmp_{i,j}$ to the judge. If a competitor $\mathbf{B}$ submits a complaint towards $\mathbf{B'}$, TTP reveals bid information of $\mathbf{B}$ to the judge and announces the re-comparison result between $y$ and $y'$ to the public. The ideal functionality $\mathcal{F}$ of a multi-party integer comparison protocol is presented in Algorithm~\ref{Alg-IdF}.

         \IncMargin{1em}
            \begin{algorithm}[t]
                \For {$i=1$ \KwTo $N$}
                   {Each competitor $\mathbf{B_i}$ sends his input $y_i$ to TTP;}
                \For {$i,j=1$ \KwTo $N,j\ne i$}
                   {TTP computes $cmp_{i,j}=1$ if $y_i\ge y_j$; otherwise, $cmp_{i,j}=0$.}

                   TTP sends to the judge:~$\{cm\!p_{i,j}|i,\!j\!\!=\!\!1,\!\dots\!,N,i\!\!\ne\!\!j\}$.

                \If {A competitor $\mathbf{B}$ submits a complaint towards $\mathbf{B'}$}
                        {TTP sends $y$ to the judge, publishes $cm\!p$ to the judge and all competitors.}
                \caption{Ideal Functionality $\mathcal{F}$ of Comparisons}
                \label{Alg-IdF}
              \end{algorithm}
              \DecMargin{1em}

        \noindent\textbf{Adversary Model.}
        It is assumed that all participants run in probabilistic polynomial time (PPT). \emph{Mixed adversary} $\mathcal{A}\!=\!(\mathcal{A}_1,\mathcal{A}_2)$ is considered corresponding to the participant roles. $\mathcal{A}_1$ and $\mathcal{A}_2$ are non-colluding and have different capabilities. $\mathcal{A}_1$ is semi-honest and can corrupts the judge to eavesdrop the input privacy of the competitors. $\mathcal{A}_2$ is either semi-honest or malicious (discussed separately) and can adaptively corrupt a subset of competitors. If $\mathcal{A}_2$ is semi-honest, a corrupted competitor can:
        1) attempt to obtain the inputs of the honest competitors but perform protocol correctly,
        2) use false input, which means that the committed input is inconsistent with the input in the ciphertexts during the comparison,
        3) submit a complaint towards the comparison result at any time,
        4) quit the comparison during execution of the scheme at arbitrary time.
        We note here $\mathcal{A}_2$, when is semi-honest, is different from the general semi-honest assumption as some malicious behaviors are allowed, like submitting false complaints, quitting the protocol at any time.

        \noindent\textbf{Network Model.}
        A weakly synchronous blockchain is used as an broadcast channel to publish messages. Each participant can timely access and put messages on the blockchain at any time. 

        Let $\mathcal P$ be a set of participants, who perform a protocol $\Pi$ to compute an ideal function $\mathcal{F}$ given inputs $In_{\mathcal P}$ and outputs $Out_{\mathcal P}$. For a corrupted subset $I$ of $\mathcal P(I\subseteq \mathcal P)$, the view of $I$ is denoted by $\mathbf{VIEW}^{\Pi}_{I}(In_{\mathcal P},Out_I,\phi_I)$, where $\phi_I$ contains messages and random numbers selected by $P_{i,i\in I}$ and related information during running $\Pi$. Let $\mathbf{Sim}_I$ be a simulator that takes the inputs of all $P_{i,i\in I}$ $(In_I)$ and the outcome of $\mathcal{F}$ received by $P_i$ $(Out_I)$ to produce a transcript of the protocol.

        \begin{definition}\label{def1}
                A protocol $\Pi$ is secure against $\mathcal{A}\!=\!(\mathcal{A}_1,\mathcal{A}_2)$, if there exist a PPT simulator $\mathbf{Sim}_I$ such that the distributions, $\mathbf{Sim}_I$ and $\mathbf{VIEW}^{\Pi}_{I}$, are computationally indistinguishable (simply denoted by symbol `$\approx_c$'), that is, $$\mathbf{Sim}_I(In_I,\mathcal{F}\!(In_\mathcal P),Out_I,\phi_I)\!\approx_c\!\! \mathbf{VIEW}^{\Pi}_{I}(In_\mathcal P,\mathcal{F}\!(In_\mathcal P),Out_I,\phi_I).$$
            \end{definition}

 \subsection{Commitment Scheme with Deposits}\label{commitment}
        In our schemes, the blockchain-based commitment scheme (CS) with deposits in \cite{ADMM} is used to bind the input with its owner, forcing the owner to follow the protocol.

        Assume a committer $\mathbf{C}$ wants to commit $y$ to a recipient $\mathbf{R}$ via blockchain.
        Initially, a hash function $\mathcal{H}$ is negotiated by the committer $\mathbf{C}$ and recipient $\mathbf{R}$.
        At the commitment phase, $\mathbf{C}$ computes $h=\mathcal{H}(S\|y)$ with some randomness $S$ as a commitment to $y$. Then he broadcasts a transaction $CS.Commit$ on blockchain revealing the commitment $h$ with some deposits $v$, which is contained in some transaction possessed by $\mathbf{C}$.
        The deposits will be redeemed by $\mathbf{C}$ during the opening phase with a transaction $CS.Open$ if he opens the correct commitment. Otherwise, the deposits will be claimed by $\mathbf{R}$ as punishment for $\mathbf{C}$.
        Afterwards, a transaction $CS.Fuse$ is created by $\mathbf{C}$ and sent to $\mathbf{R}$ with his signature on it enabling $\mathbf{R}$ to claim the deposits.

        At the opening phase, $\mathbf{C}$ reveals $S$ and $y$ via a transaction $CS.Open$ to redeem his deposits. If the transaction $CS.Open$ is not recorded on the blockchain after a deadline, then $\mathbf{R}$ broadcasts the transaction $CS.Fuse$ to gain the deposits of $\mathbf{C}$, or $\mathbf{R}$ learns $S$ and $y$ if $CS.Open$ appears on the blockchain then $\mathbf{C}$ redeems his deposits.

        For an honest committer, the deposits will be redeemed after opening the commitment. If the committer can not correctly open the commitment, then the recipient will claim the deposits. CS satisfies the $binding$ property of a commitment scheme since hash function is collision-resistant while $hiding$ is guaranteed by unpredictability of the hash function.

\section{Three-Party Integer Comparison Protocol}\label{sect-TPIC-scheme}
        In this section, two TPIC protocols are given against semi-honest and malicious $\mathcal{A}_2$, respectively.

  \subsection{Semi-honest Competitors}
        In this subsection, a secure TPIC against semi-honest adversary $\mathcal{A}_2$ is given. We improve the two-party integer comparison protocol in \cite{CEK} and give security proofs. Assume that two competitors, $P_1$ and $P_2$, holding $\alpha$ and $\beta$, respectively, want to prove to a judge $\mathbf{A}$ which one of $\alpha$ and $\beta$ is larger without uncovering them.

        The strategy involves three homomorphic public key encryption (PKE) schemes, by which to do computations under ciphertexts:
        1)$PKE_1\!=\!(\textbf{Gen}_1, \textbf{Enc}_1, \textbf{Dec}_1)$ from \cite{CEK} is semantically secure with threshold property that if $m_1\!+\!m_2\!<\!d$, then $\textbf{Enc}_1(m_1)^{b^{m_2}}\!=\!\textbf{Enc}_1(m_1\!+\!m_2)$, otherwise, $\textbf{Enc}_1(m_1)^{b^{m_2}}\!=\!\textbf{Enc}_1(0)$.
        2)$PKE_2\!=\!(\textbf{Gen}_2, \textbf{Enc}_2, \textbf{Dec}_2)$ with message space $\mathcal{M}_{2}$ is multiplicatively homomorphic.
        3)$PKE_3\!=\! (\textbf{Gen}_{\oplus}, \textbf{Enc}_{\oplus}, \textbf{Dec}_{\oplus})$ is semantically secure and additively homomorphic with message space $\mathcal{M}_{\oplus}$ satisfying $\textbf{Enc}_\oplus(\mathcal{M}_{\oplus})\subseteq\mathcal{M}_2$.
        An efficient candidate for $PKE_3$ is \cite{ST} and proper candidates for $PKE_2$ can be RSA or ElGamal encryption.

        Let
        $\alpha=\alpha_{k-1}d^{k-1}+\alpha_{k-2}d^{k-2}+\cdots+\alpha_1 d+\alpha_0$ and $\beta=\beta_{k-1}d^{k-1}+\beta_{k-2}d^{k-2}+\cdots+\beta_1d+\beta_0$, where $0\le\alpha_i,\beta_i<d$.
        If exactly one of the $k$ Boolean expressions below is \textbf{True}:
            $(\alpha_{k-1}>\beta_{k-1})$,
            $(\alpha_{k-1}=\beta_{k-1})\wedge(\alpha_{k-2}>\beta_{k-2})$,
            $(\alpha_{k-1}=\beta_{k-1})\wedge(\alpha_{k-2}=\beta_{k-2})\wedge(\alpha_{k-3}>\beta_{k-3})$,
            $\cdots$,
            $(\alpha_{k-1}=\beta_{k-1})\wedge(\alpha_{k-2}=\beta_{k-2})\wedge\cdots\wedge(\alpha_0>\beta_0)$,
        then $\alpha>\beta$; otherwise $\alpha\le\beta$.

        Specifically,
        $P_i(i=1,2)$ runs $\mathbf{Gen}_i(1^\lambda)$ to generate the key pair $(\mathcal{PK}_i,\mathcal{SK}_i)$ while the judge $\mathbf{A}$ runs $\textbf{Gen}_{\oplus}(1^\lambda)$ to get the key pair $(\mathcal{PK}_3,\mathcal{SK}_3)$. They publish their public keys $\mathcal{PK}_i$ to each other where $\mathcal{PK}_1=(n,b,d,g,h,u)$. In the following, let $\ell=k-1,\dots,0$ and the notation $x\leftarrow_\$S$ denotes that a value $x$ is chosen uniformly at random from a set $S$.
        \begin{enumerate}
            \item $P_1$ selects $r_{1,\ell}\!\!\leftarrow_\$\!\!\{1,\dots,2^u\!-\!1\}$ and computes $C_\ell=g^{b^{(d-\alpha_\ell-1)}}h^{r_{1,\ell}}$, then he publishes $C$ on blockchain, where $C\!\!=\!\!(C_{k-1}$, $\dots,C_0)$.
            \item $P_2$ selects $r_{2,\ell}\!\!\leftarrow_\$\!\!\{1,\dots,2^u\!-\!1\}$, $s_\ell\!\!\leftarrow_\$\!\!\{1,\dots,b^{d}\!-\!1\}$, $s.t.,s_\ell\!\!\not\equiv\!\!0\bmod b$ and computes $D_\ell\!\!=\!\!C_\ell^{b^{\beta_\ell}}g^{s_\ell}h^{r_{2,\ell}}$; $A_{2,k-1}=\textbf{Enc}_2(\textbf{Enc}_{\oplus}(s_{k-1}))$,
                $A_{2,k-2}\!\!=\!\!\textbf{Enc}_2(\textbf{Enc}_{\oplus}(\beta_{k-1}\|s_{k-2}))$,
                $\dots$,
                $A_{2,0}\!\!=\!\!\textbf{Enc}_2(\textbf{Enc}_{\oplus}(\beta_{k-1}\|\beta_{k-2}\|\cdots\|s_{0}))$. She puts $(D$, $A_2)$ on blockchain, where $D=(D_{k-1},\dots,D_0)$ and $A_2=(A_{2,k-1}$, $\dots,A_{2,0})$.
            \item $P_1$ computes $w_\ell\!\!=\!\!\textbf{Dec}_1(D_\ell)$; $A'_{1,k-1}\!\!=\!\!\textbf{Enc}_2(\textbf{Enc}_{\oplus}(-w_{k-1}))\cdot$ $ A_{2,k-1}$, $A'_{1,k-2}=\textbf{Enc}_2(\textbf{Enc}_{\oplus}(-\\(\alpha_{k-1}\|w_{k-2})))\cdot A_{2,k-2}$, \dots,
                $A'_{1,0} = \textbf{Enc}_2(\textbf{Enc}_{\oplus}(-(\alpha_{k-1}\|\alpha_{k-2}\|\cdots\|w_{0})))\cdot A_{2,0}$.
                 Then he blinds $A_{1,\ell}'$ by computing $A_{1,\ell}\!=\!(A'_{1,\ell})^{r_\ell}$ for $0\!\ne\! r_\ell$ $\leftarrow_\$\mathcal{M}_{\oplus}$. He gets $A_1=(A_{1,k-1'},\dots,A_{1,0'})$ by shuffling $(A_{1,k-1},\dots,A_{1,0})$ via a random permutation $\pi_1$. He puts $A_1$ on blockchain.
            \item $P_2$ computes $A_{0,\ell}'=\textbf{Dec}_2(A_{1,\ell})$ and blinds $A_{0,\ell}'$ by computing $A_{0,\ell}=(A'_{0,\ell})^{r_\ell'}$ for $0\ne r_\ell'\leftarrow_\$\mathcal{M}_{\oplus}$. Then, she shuffles $(A_{0,k-1},\dots,A_{0,0})$ to $A_0=(A_{0,k-1'},\dots,A_{0,0'})$ by a random permutation $\pi_2$. She publishes $A_0$ on blockchain.
            \item $\mathbf{A}$ computes $m_\ell=\textbf{Dec}_{\oplus}(A_{0,\ell})$. If for all $\ell$, $m_\ell\ne0$, then, $\alpha\ge\beta$; otherwise, $\alpha\!<\!\beta$.
        \end{enumerate}

        Note that the string $\beta_{k-1}\|\beta_{k-2}\|\cdots\|s_{0}$ is regarded as an integer to be encrypted, which can be achieved by choosing appropriate message spaces.

        \textbf{Correctness.}
        $P_1$ computes $w_\ell\!\!=\!\!(b^{d-\alpha_\ell-1+\beta_\ell}\!+\!s_\ell)\bmod b^d$. If $\beta_\ell\!\!\ge\!\!\alpha_\ell+1$ holds, then $w_\ell\!\!=\!\!s_\ell$.
        Since the encryptions, $PKE_2$ and $PKE_3$, are homomorphic, we have
        $A_{1,\ell}\!\!=\!\!\textbf{Enc}_2[(\textbf{Enc}_{\oplus}(\beta_{k-1}\|\!\cdots\!\|\beta_{\ell+1}\|s_{\ell})\cdot$ $\textbf{Enc}_{\oplus}(\!-(\alpha_{k-1}\|\!\!\cdots\!\!\|\alpha_{\ell+1}\| w_{\ell})))^{r_\ell}\!]$
        and
        $A_{0,\ell}\!\!=\!\!\textbf{Enc}_{\oplus}[r_\ell r_\ell'(\beta_{k-1}\|\!\cdots\!\|\beta_{\ell+1}\|s_{\ell}\!-\!\alpha_{k-1}\|\cdots\|\alpha_{\ell+1}\|w_{\ell})]$.
        If there exists only one $\ell$, $s.t.$, $m_\ell=0$, then
        $(\alpha_{k-1}=\beta_{k-1})\wedge(\alpha_{k-2}=\beta_{k-2})\wedge\cdots\wedge(\alpha_{\ell+1}=\beta_{\ell+1})\wedge(w_\ell=s_\ell)$ holds, and the $\ell$-th Boolean expressions is \textbf{True}, $i.e.$,
        $\alpha<\beta$. Otherwise, if $m_\ell\ne0$ for all $\ell$, then $\alpha\ge\beta$ holds.

        \begin{theorem}\label{thm1}
                 Assume that the three PKEs are secure. Then, the above TPIC protocol is secure against mixed adversary $\mathcal{A}\!=\!(\!\mathcal{A}_1,\mathcal{A}_2\!)$, where $\mathcal{A}_1$ and $\mathcal{A}_2$ are non-colluding and semi-honest, $\mathcal{A}_1$ might corrupt the judge $\mathbf{A}$, $\mathcal{A}_2$ might corrupt a competitor $P_{\!1}$ or $P_{\!2}$. I.e., the privacy, $\alpha$ and $\beta$, of $P_{\!1}$ and $P_{\!2}$ respectively, is preserved.
            \end{theorem}
        \begin{proof}[Sketch Proof]
          By Definition~\ref{def1}, we construct simulators for the three PPT parties. Here $\Pi$ is TPIC, and the ideal functionality $\mathcal{F}$ computes the relationship between $\alpha$ and $\beta$, $i.e., \mathcal F(\alpha,\beta)=1,$ if $\alpha>\beta$, otherwise $\mathcal F(\alpha,\beta)=0$.

        Given the public keys, $\mathcal{SK}_1$ and $\alpha$,
        $\mathbf{Sim}_1$ encrypts $\alpha$ by $PKE_1$ as $C$ which is identical to the real. Unknowing $\beta$, $\mathbf{Sim}_1$ just encrypts $0^{|C|}$ via $PKE_1$ to get $D_s$. Since $PKE_1$ is semantically secure, $D_s$ is computationally indistinguishable from $D$ in the real.
        $\mathbf{Sim}_1$ also encrypts numbers randomly chosen from $\mathcal{M}_{\oplus}$ to obtain $A_2$ which is indistinguishable from the real due to the (semantical) security of $PKE_3$ and $PKE_2$. Then $\mathbf{Sim}_1$ computes $A_1$ as the real execution, which is also computationally indistinguishable. To simulate $A_0$, $\mathbf{Sim}_1$ encrypts random numbers chosen from $\mathcal{M}_{\oplus}$ by semantically secure $PKE_3$. Thus $A_0$ is also computationally indistinguishable. So we have $\mathbf{Sim}_1(\alpha,\mathcal{SK}_1,\phi_1)\!\!\approx_c\!\! \mathbf{VIEW}^{\Pi}_{1}((\alpha,\beta),\mathcal{SK}_1,\phi_1)$. Note that $P_1$ has no output.

        $\mathbf{Sim}_2$ and $\mathbf{Sim}_A$ can be constructed similarly to $\mathbf{Sim}_1$. Note that for $\mathbf{Sim}_A$ given the comparison result $\mathcal{F}\!(\alpha,\beta)$, $A_0$ should be consistent with the result. If $\alpha\!<\!\beta$, there should exist only one $\ell,s.t.,\textbf{Dec}_{\oplus}(A_{0,\ell})\!=\!0$. In this case, $\mathbf{Sim}_A$ encrypts 0 and $(\ell\!\!-\!\!1)$ non-zero random numbers by $PKE_3$, then shuffles them. Otherwise, $\mathbf{Sim}_A$ encrypts $\ell$ non-zero random numbers then shuffles them. Similarly, we have $\mathbf{Sim}_2(\beta,\mathcal{SK}_2,\phi_2)\!\!\approx_c\!\!\mathbf{VIEW}^{\Pi}_{2}((\alpha,\beta),\mathcal{SK}_2,\phi_2)$ for $P_2$ with no output, and $\mathbf{Sim}_A(\mathcal{F}(\alpha,\beta),\mathcal{SK}_3,Out_A,\phi_3)\approx_c \mathbf{VIEW}^{\Pi}_{1}(\mathcal{F}(\alpha,\beta),\mathcal{SK}_3,Out_A,\phi_3)$ for $\mathbf{A}$ who outputs the comparison result.
        \end{proof}

   \subsection{Malicious Competitors}
        Considering malicious $\mathcal{A}_2$ who may corrupt $P_1$ or $P_2$, we improve the comparison protocol of Fischlin \cite{Fisch}, where Goldwasser-Micali (GM) \cite{GM82} encryption is used. The public key is $pk=(n,z)$ while the private key is $sk\!=\!\frac{(p\!-\!1)(q\!-\!1)}{4}$, where $n\!=\!pq$, $p\!\equiv\!q\!\equiv\!3\bmod 4$, $z\!\equiv\!-1\bmod n$. One can encrypt a bit $b$ by computing $c\!=\!r^2z^b\!\bmod\!n$ with randomly chosen $r$ from $Z_n^\ast$ and decrypt $c$ by computing $1\!-\!c^{sk}\!\bmod\!n$. GM is homomorphic for encryptions of two bits $b_1$ and $b_2$: $Enc(b_1)\cdot Enc(b_2)\!=\!Enc(b_1\oplus b_2)$ (\emph{XOR}), $Enc(b_1)\cdot z\!=\!Enc(1\!-\!b_1)$ (\emph{NOT}). For a GM ciphertext $c$, re-encryption is $ReEnc(c)\!=\!c\!\cdot\!Enc(0)$.

        \emph{AND-homomorphic}. GM encryption can be modified to support Boolean AND homomorphic following \cite{SYY99}, by which, a single bit is encrypted to $\iota$-element sequence, where $\iota$ is the soundness parameter. $Enc^{AND}(1)\!=\!(Enc(0),\dots,Enc(0))$ and $Enc^{AND}(0)\!=\!(Enc(a_1), \dots, Enc(a_\iota))$  for $\iota$ randomly chosen bits. The AND-decryption outputs $1$ if each element in the sequence is decrypted to $0$; otherwise, it  outputs $0$. This decryption is correct with probability $1\!-\!2^{-\iota}$. For two ciphertexts $Enc^{AND}(b_1)\!=\!(c_1,\dots,c_\iota)$ and $Enc^{AND}(b_2)\!=\!(c'_1,\dots,c'_\iota)$,  $Enc^{AND}(b_1)\!\cdot\!Enc^{AND}(b_2)\!=\!(c_1\!\cdot\!c'_1,\dots,c_\iota\!\cdot\!c'_\iota)\!=\!Enc^{AND}(b_1\!\wedge\!b_2)$ (\emph{AND}).
        Re-encryption for AND-encryption is defined as $ReEnc^{AND}(c_1,\dots,c_\iota)\!=\!(ReEnc(c_1),\dots,ReEnc(c_\iota))$.

        \emph{AND-embedding}. An existing GM ciphertext $c$ of bit $b$ can be embedded to AND-ciphertext $Enc^{AND}(b)\!=\!(c_1,\dots,c_\iota)$ without decryption. For $\iota$ randomly chosen bits $a_1,\dots,a_\iota$, if $a_i\!=\!1$, then set $c_i\!=\!Enc(0)$; otherwise, set $c_i\!=\!Enc(0)\!\cdot\! c\!\cdot\! z\bmod n$. Such embedding is correct with probability $1\!-\!2^{-\iota}$.

        Fischlin \cite{Fisch} proposed secure protocol to compare two integers against semi-honest adversaries, which is a bit-wise comparison. Given $\eta$-bit integers $\alpha$ and $\beta$ with bit representations $\alpha\!=\!\alpha_1\dots\alpha_\eta$ and $\beta\!=\!\beta_1,\dots,\beta_\eta$, one can compute $\alpha\!>\!\beta$ by evaluating Boolean circuit $F\!=\!\bigvee_{\ell=1}^\eta(\alpha_\ell\!\wedge\!\neg\beta_\ell\wedge\bigwedge_{u=\ell+1}^\eta(\alpha_u\!=\!\beta_u)).$
        $F\!=\!1$ if and only if $\alpha\!>\!\beta$. $\bigvee_{\ell=1}^\eta$ is XOR operation and exactly one term will be 1 if $\alpha\!>\!\beta$; otherwise, all terms will be 0. Moreover, $(\alpha_u\!=\!\beta_u)$ equals $\neg(\alpha_u\!\oplus\!\beta_u)$. Thus, $F$ can be homomorphically evaluated using GM encryption.

        We improve the protocol of Fischlin \cite{Fisch} to be a secure three-party scheme against malicious $P_{\!1}$ and $P_2$. $\mathbf{A}$ gets the comparison result, while $P_{\!1}$ and $P_{\!2}$ protect their private inputs $\alpha$ and $\beta$. It is assumed that $\mathbf{A}$ will not collude with $P_1$ or $P_2$.
        Except for GM encryption, RSA is also used in our scheme. $\mathbf{A}$ generates a GM key pair and publishes the public key before the comparison. For $i\!=\!1,2$, $P_{i}$ generates RSA public/ private key pair $(e_i,d_i)$. The private inputs are first encrypted by GM, then encrypted by RSA, such double encryption is called  GM-RSA encryption in the following. Since, RSA is multiplicatively homomorphic, the homomorphic properties (NOT,XOR,AND) of GM will not be affected involving only multiplication. In the following, encrypting a sequence means encrypting each element separately.
            \begin{enumerate}
              \item $P_i$ bit-wisely encrypts its private input by GM into $C_i\!=\!(C_{i,1},\dots,C_{i,\eta})$, then, encrypts $C_i$ by RSA into $D_i\!=\!(C_i)^{e_i}$. To prove the knowledge of the private input, a non-interactive zero-knowledge (NIZK) proof $\pi_i$ given below is added. Finally, $P_i$ publishes $\langle e_i,D_i,\pi_i\rangle$ on blockchain.
              \item $P_2$ encrypts $C_2$ into $D_{2,1}\!=\!C_2^{e_1}$ using $e_1$ of $P_1$ with a NIZK $\hat{\pi}$ proving $C_2$ is contained in $D_2$ and $D_{2,1}$. $P_2$ homomorphically computes all $\neg(\alpha_u\!\oplus\!\beta_u)$ and $\neg\beta_\ell$. $P_2$ embeds $C_1$ and $C_2$ into AND-homomorphic GM by manipulating ciphertexts $D_1,D_{2,1}$. Specifically, $P_2$ selects random coins and perform GM-RSA encryption on them, which then multiplied by $D_1,D_{2,1}$. $P_2$ also attaches knowledge proofs $\pi_{2,1}$ of the random coins following the construction of $\pi$ below. For each $\ell\! =\!1,\dots,\eta$, $P_2$ then computes a ciphertext $\bar{c}_\ell$ of $\bar{b}\!=\!\alpha_\ell\!\wedge\!\neg\beta_\ell\wedge\bigwedge_{u=\ell+1}^\eta(\alpha_u\!=\!\beta_u)$.
                  $P_2$ randomly shuffles sequence $(\bar{c}_1,\dots,\bar{c}_\eta)$ to $\overline{res}_{2,1}\!=\!shuffle(\bar{c}_1,\dots,\bar{c}_\eta)$ by a random permutation $shuffle$ with a NIZK proof $P_{2,1}^{shuffle}$ of $shuffle$ following \cite{OKST97,RW04}. Let proof $P_{2,1}^{eval}\!=\!(D_{2,1},\hat{\pi},\pi_{2,1},P_{2,1}^{shuffle})$. $P_2$ encrypts $P_{2,1}^{eval}$ into $[P_{2,1}^{eval}]$ by the public key of $\mathbf{A}$ such that the proof is visible only to $\mathbf{A}$.
                  $P_2$ puts $\langle \overline{res}_{2,1},[P_{2,1}^{eval}] \rangle$ on blockchain. $P_{\!1}$ performs operations similarly to $P_2$ and publishes $\langle \overline{res}_{1,2},[P_{1,2}^{eval}] \rangle$.
              \item $\mathbf{A}$ verifies $P_{2,1}^{eval},P_{1,2}^{eval}$ by checking the NIZK proofs. Then,            $\mathbf{A}$ published the verification results on blockchain.
              \item $P_1$ decrypts $\overline{res}_{2,1}$ by $d_1$ into $res_{2,1}$ and puts $\langle res_{2,1}\rangle$ on the blockchain. Similarly, $P_2$ publishes $\langle res_{1,2}\rangle$.
              \item $\mathbf{A}$ can verify correctness of $res_{2,1}$ by re-encrypting (via $e_1$) before decrypting. If exactly one $1$ is contained, then $\alpha\!>\!\beta$. Otherwise, $\mathbf{A}$ decrypts $res_{1,2}$. If exactly one $1$ is obtained, then $\alpha\!<\!\beta$. Otherwise $\alpha\!=\!\beta$.

            \end{enumerate}
        \emph{NIZK-$\pi$}. For a GM public key $(n,z)$, a RSA public key $(e,N)$, a GM-RSA ciphertext $C\!=\!c^e\!=\!(r^2z^b\!\bmod\!n)^e\!\bmod\! N$ of a bit $b$, a prover can prove the knowledge of $b$ in $\kappa$ rounds as in \cite{Katz03}. The prover randomly chooses $\kappa$ numbers $\{r_i\}_{i=\!1}^\kappa$, and computes $\{A_i\!=\!r_i^{4e}\}_{i=\!1}^\kappa$. With the help of a publicly known random oracle $\mathcal H\!:\!\{0,1\}^\ast\!\to\!\{0,1\}^\kappa$, the prover computes challenges $\gamma\!=\!\mathcal H(A_1,\dots,A_\kappa)$ regarded as a bit string. The prover concludes the proof by sending responses $(\gamma,\{R_i\!=\!(r^{\gamma_i}\!\cdot\!r_i)^e\}_{i=\!1}^\kappa)$. A verifier computes all $A_i'\!=\!\frac{R_i^4}{C^{2\gamma_i}}$ and accepts if $\gamma\!=\!\mathcal H(A_1',\dots,A_\kappa')$. The security relies on the security of RSA and the original proof of knowledge $b$ for GM ciphertext in \cite{Katz03}.

        \noindent\emph{NIZK-$\hat{\pi}$}. For two RSA public keys $e_1,e_2$, public ciphertexts $X_1,X_2$, a prover can prove the knowledge of $x$ such that $X_1\!=\!x^{e_1}\wedge X_2\!=\!x^{e_2}$. The prover randomly chooses $y$ and computes $Y_1\!=\!y^{e_1},Y_2\!=\!y^{e_2}$. Then, the prover computes challenges $\gamma\!=\!\mathcal H(Y_1,Y_2)$ regarded as an integer. The response is $(\gamma,yx^\gamma)$.  A verifier computes $Y_1'\!=\!\frac{(yx^\gamma)^{e_1}}{X_1^{\gamma}}$, $Y_2'\!=\!\frac{(yx^\gamma)^{e_2}}{X_2^{\gamma}}$, and accepts if $\gamma\!=\!\mathcal H(Y_1',Y_2')$. For security of $\hat{\pi}$, one can refer to \cite{GQ88-RSA-zk} for proving the knowledge of $x$ such that $X\!=\!x^e$.

            \begin{theorem}\label{thm2}
                 The above TPIC protocol is secure against mixed adversary $\mathcal{A}\!=\!(\!\mathcal{A}_1,\mathcal{A}_2\!)$, where $\mathcal{A}_1$ is semi-honest and can corrupt the judge $\mathbf{A}$, $\mathcal{A}_2$ is malicious and can corrupt $P_{\!1}$ or $P_{\!2}$. I.e., the privacy, $\alpha$ and $\beta$, of $P_{\!1}$ and $P_{\!2}$ respectively, is preserved.
            \end{theorem}
        \begin{proof}[Sketch Proof]
          Following Definition~\ref{def1}, we construct simulators for the three PPT parties with the help of the ideal functionality $\mathcal{F}$, which computes the relationship between $\alpha$ and $\beta$, $i.e., \mathcal F(\alpha,\beta)=1,$ if $\alpha>\beta$, otherwise $\mathcal F(\alpha,\beta)=0$. We denote by $\mathbf{Sim}^{\pi}$ and $\mathbf{Sim}^{\hat{\pi}}$ the simulators of NIZK $\pi$ and $\hat{\pi}$.

          Given $(e_1,d_1)$ and $\alpha$, $\mathbf{Sim}_1$ computes $\langle e_1,D_1,\pi_1\rangle$ as in the real. Then $\mathbf{Sim}_1$ invokes $\mathbf{Sim}^{\pi}$ to simulate behavior of $P_2$. This simulation is computationally indistinguishable from the real.
          In step 2, $\mathbf{Sim}_1$ computes ciphertexts with the help of TTP and $\mathbf{Sim}^{\hat{\pi}}$. Since the final ciphertexts $\overline{res}_{1,2}$ will reveal the relationship between $\alpha$ and $\beta$, $\mathbf{Sim}_1$ can just encrypts random numbers according to $F(\alpha,\beta)$. For $D_{1,2}$, if $F(\alpha,\beta)\!=\!1$, $\mathbf{Sim}_1$ encrypts only one $1$ and 0 for other $\eta\!-\!1$ elements; otherwise, it always encrypts 0 for all $\eta$ elements. Afterwards, it invokes $\mathbf{Sim}^{\hat{\pi}}$ to guarantee NIZK is acceptable.
          $\mathbf{Sim}_1$ encrypts random bits and gives NIZK following $\pi$ to embed $D_2,D_{1,2}$ to AND-homomorphic GM.
          Then it invokes the simulator for $shuffle$ to shuffle the sequence. This simulation process is also computationally indistinguishable from the real computation.
          In step 4, $\mathbf{Sim}_1$ does decryptions as in the real with the given private key $d_1$, which is identical to the real. One can construct $\mathbf{Sim}_2$ as same as the construction of $\mathbf{Sim}_1$.

          The judge $\mathbf{A}$, who has no private input, only obtains RSA ciphertexts in step 1,2,3 and NIZK proofs. $\mathbf{A}$ can revoke $\mathbf{Sim}^{\pi}$ and $\mathbf{Sim}^{\hat{\pi}}$. With the help of TTP, the simulation should be consistent with  $F(\alpha,\beta)$ as discussed for $\mathbf{Sim}_1$ in step 2, such that the resulting simulator is computationally indistinguishable from the real execution.
        \end{proof}

   \section{Applications}\label{sect-Appli}
        In this section, the applications based on three-party comparisons are given with the help of blockchain.

 \subsection{Joint bidding}
        The proposed comparison protocols can be used for joint bidding \cite{CL02}, which is a multi-party scenario.
        Joint bidding is a form of alliance and is used extensively in the construction and insurance industries, allowing a group of firms to collectively undertake a big project beyond any single. The logic is to take advantage of the collective size of the group. To choose proper partners, the project owner needs to check their financial assets which they do not like to uncover.

        A solution is given and instantiated by the three-party comparison for malicious competitors.
        In our design, all the firms $\mathbf{B_i}$ run parallelized TPIC with the project owner $\mathbf{A}$. At last, $\mathbf{A}$ learns the order of all assets to select needed partners. If the owner wants to choose partners possessing assets more that a prefer value $x$, it can play a role of firm to compare all assets with $x$. $\mathbf{A}$ then chooses firms whose assets is larger than $x$ whilst $x$ is not known to them.

        At the beginning, $\mathbf{A}$ broadcasts the GM public key, all firms perform the following scheme.
            \begin{enumerate}
              \item Each firm $\mathbf{B_i}$ publishes $\langle e_i,D_i,\pi_i\rangle$ on blockchain committing to assets $y_i$. If a prefer value $x$ is needed, $\mathbf{A}$ publishes $\langle e,D,\pi\rangle$ on chain committing to $x$.
              \item Each $\mathbf{B_i}$ computes $\langle \overline{res}_{i,j},[P_{i,j}^{eval}] \rangle$ for all $j\!\ne\!i$ and puts them on blockchain.
              \item $\mathbf{A}$ verifies all NIZK proofs and publishes the verification results.
              \item $\mathbf{B_i}$ decrypts all $\overline{res}_{j,i}$ into $res_{j,i}$ and publishes them on blockchain.
              \item $\mathbf{A}$ decrypts $res_{i,j}$ getting the order of all $y_i$ (and $x$). Then $\mathbf{A}$ can select proper partners according to the ranking of assets (larger than $x$).
            \end{enumerate}
         Here, the commitment with deposits can also be used to incentivize honest performance.

         The security is guaranteed by theorem~\ref{thm2}. One can refer to \cite{BK2} for the time cost for GM encryption and related NIZK proofs. For primes $|p|\!=\!|q|\!=\!768$, 32-bit integers, $\kappa\!=\!\iota\!=\!40$ for soundness error $2^{-40}$, the total time is less than 1s. To guarantee RSA is compatible with GM, we test encryption and decryption time setting $|p|\!=\!|q|\!=\!1152$ on PC (AMD PRO A10-8770 R7, 10 COMPUTE CORES 4C+6G) running at 3.5GHz. The total RSA encryption time for $1280(=\eta\iota)$ times is less than 200ms while 6s for decryption. The time cost is acceptable in the real.

         The scheme runs in 4 blocks and the per-firm communication cost in each step is given in table~\ref{table_comm-cost}, where $\eta$ is bit length for the asset value, $\iota$ is the soundness parameter in AND homomorphic of GM,  $\kappa$ is the soundness parameter in NIZK proofs, $|p|$ is the prime size for RSA, $N$ is the number of firms. The communication coat in step 2 does not contain the ciphertexts $[P_{i,j}^{eval}]$. $P_{i,j}^{eval}$ can be encrypted by an alternative scheme instead of GM since GM is a bit-wise encryption, thus consumes much communication.

        \doublerulesep 0.1pt
      \begin{table}[t]
      \begin{footnotesize}
      \caption{Per-firm communication cost.}
      \label{table_comm-cost}
         \begin{center}
                \begin{tabular}{ccccc}
                    \toprule
                    \ --   \quad & step 1            \quad &step 2     \quad &step 4 \\
                    \midrule
                    comm(bits)\quad &$2(\eta\!+\!\kappa)|p|$ \quad&$2(N\!-\!1)\eta\iota|p|$ \quad&$2(N\!-\!1)\eta\iota|p|$ \\
                \bottomrule
                \end{tabular}
         \end{center}
      \end{footnotesize}
      \end{table}

 \subsection{Auctions from TPIC}\label{multi-party-case}
        In this subsection, a general version of three-party comparison is presented by a privacy-preserving auction scheme running in a constant number of rounds. In an auction scenario, $N$ bidders submit bids hoping to win without revealing the bid information while the auctioneer expects to determine a winner with the highest bid. Traditional seal-bid scheme heavily relies on the auctioneer $\mathbf{A}$ who is assumed to be trusted since $\mathbf{A}$ learns all bid information. This dependence will however lead to unfairness as $\mathbf{A}$ may collude with a bidder or leak the bids of the honest to a malicious bidder. On the other hand, a bidder may perform badly to influence the outcome of the auction.
        In our instantiation, the commitment scheme (CS) with deposits is used to bind bid as well as the ciphertexts with each bidder. The deposits are used to punish bidders with bad behavior like inconsistent commitment or wrong complaint. The three-party protocol for semi-honest competitors is used to compare all bids under ciphertexts to preserve privacy.

        Roughly, each bidder $\mathbf{B_i}$ is required to run CS to broadcast $CS.Commit_i$ on blockchain with some deposits committing to his bid $y_i$ and its ciphertexts. He is also supposed to provide a transaction $CS.Fuse_i$ for the auctioneer $\mathbf{A}$ with his signature on it. Then, $\mathbf{B_i}$ runs TPIC to compare his bid with that of $\mathbf{B_j}$ for all $j\!\ne\!i$ in parallel.
        At the last round of TPIC, $\mathbf{A}$ gets all the comparison results and announces the highest bid holder to be the auction winner.
        The determined winner $\mathbf{B_w}$ should open his commitment on blockchain and other bidders can verify its soundness.
        A bidder $\mathbf{B_i}$ who loses the auction can redeem the deposits by $CS.Open_i$ which will disclose the bid information, thus an alternative $CS.Refund_i$ is used to return the deposits. $CS.Refund_i$ contains no bid information and is valid only signed by both $\mathbf{B_i}$ and $\mathbf{A}$.
         If a complainer $\mathbf{B}$ submits a complaint towards $\mathbf{B'}$, $\mathbf{B}$ is required to open his commitment to $\mathbf{A}$ for consistency verification.
         Then they run TPIC again, after which $\mathbf{A}$ gets the result and penalizes the loser by broadcasting the transaction $CS.Fuse$ on blockchain to confiscate his deposits before exclusion. The punishment is necessary to prevent bid information extraction from malicious comparisons by continuous complaint submission, because each comparison will reveal a lower or upper bound of the bid.
        Finally, $\mathbf{A}$ signs the transaction $CS.Refund$ from the honest bidders and broadcasts them on blockchain to return their deposits.

        Specifically, all bidders and the auctioneer negotiate a hash function $\mathcal{H}$ and the deposits value used in the commitment. Each $\mathbf{B_i}$ runs both $\mathbf{Gen}_1(1^\lambda)$ and $\mathbf{Gen}_2(1^\lambda)$ to generate the key pairs $(\mathcal{PK}^1_i,\mathcal{SK}^1_i)$ and $(\mathcal{PK}^2_i,\mathcal{SK}^2_i)$ while $\mathbf{A}$ runs $\textbf{Gen}_{\oplus}(1^\lambda)$ to get the key pair $(\mathcal{PK}_3,\mathcal{SK}_3)$. They all publish their public keys, and
        then run the following parallelized TPIC, where $N$ is the number of bidders, $i,j\!\in\!\{1,2,\dots,N\},i\!\ne\!j$, and $m_{i,j}$ denotes that the message $m$ in TPIC is computed by $\mathbf{B_i}$ and transferred to the target receiver $\mathbf{B_j}$.
            \begin{enumerate}
              \item Each $\mathbf{B_i}$ encrypts $y_i$ into $C_i$ via $PKE_1$, and computes  $H_{y_i}^{C_i}\!=\!\mathcal{H}(y_i\|C_i\|S_i)$ with a random string $S\!_i$ as his commitment.
                  $\mathbf{B_i}$ creates two transactions $(CS.Fuse_i,CS.Refund_i)$ and encrypts them by $\mathcal{PK}_{\!3}$ of $\mathbf{A}$ into $[CS.Fuse_i,CS.Refund_i]$.
                  Finally, $P_{\!1}$ broadcasts
                  $$\langle H_{y_i}^{C_i},C_i,[CS.Fuse_1,CS.Refund_1]\rangle$$
                  on blockchain by a transaction $CS.Commit_i$ with some deposits;
              \item $\mathbf{B_i}$ computes $(D,A_2)_{i,j}$ for all $j\ne i$ and puts them on blockchain;
              \item $\mathbf{B_i}$ computes $(A_1)_{i,j}$ for all $j\ne i$ and publishes them on blockchain;
              \item $\mathbf{B_i}$ computes $(A_0)_{i,j}$ for all $j\ne i$ and posts them on blockchain;
              \item For $\ell=k-1,\dots,0$, $\mathbf{A}$ computes $(m_\ell)_{i,j}=\textbf{Dec}_{\oplus}((A_{0,\ell})_{i,j})$.
                  If $\exists i$ for $\forall j\ne i,\forall\ell$, $(m_\ell)_{i,j}\ne0$,
                        then $\mathbf{A}$ determines $\mathbf{B_i}$ to be the winner who is required to open his commitment for public verification.
              \item If a complainer $\mathbf{B}$ submits a complaint towards $\mathbf{B'}$ on blockchain, then
                        he is required to reveal $y\|C\|S$ to $\mathbf{A}$ for consistency verification. If it is not consistent, $\mathbf{A}$ penalizes $\mathbf{B}$ by broadcasting his $CS.Fuse$ on blockchain confiscating his deposits. Otherwise, they run TPIC with the complained bidder, after which, $\mathbf{A}$ gets the result and penalizes the loser who is excluded from now on.
                    Finally, $\mathbf{A}$ signs the transaction $CS.Refund$ from the honest bidders and broadcasts them on blockchain to return their deposits.
            \end{enumerate}
        In this scheme, we assume that there exists only one bidder holding the highest bid. If more than one bidders submit the same highest bid, then they can run TPIC to bid again.

                \begin{theorem}\label{thm3}
                        Our auction scheme is secure against mixed adversary $(\mathcal{A}_1,\mathcal{A}_2)$, where $\mathcal{A}_1$ may corrupt the auctioneer $\mathbf{A}$, $\mathcal{A}_2$ may corrupt one or more bidders. I.e., the bid privacy of honest bidders is preserved.
                \end{theorem}

         \begin{proof}[Sketch Proof]
                  By the ideal functionality $\mathcal{F}$, we construct a simulator for each party. For each bidder $\mathbf{B_i}$, we construct $\mathbf{Sim}_i$ given the bid $y_i$, related public and private keys, and the output $\mathcal{F}_i$ of $\mathbf{B_i}$ receiving from TTP.
                  Since the commitments is based on hash function, each simulator can hash a random string as simulated commitment for others which is computationally indistinguishable from the real one.
                  In step 1, given $y_i$, $\mathbf{Sim}_i$ directly computes $C_i$, $CS.Commit_i$, $CS.Fuse_i$ and $CS.Refund_i$. These are identical to the real.
                  In step 2, $\mathbf{B_i}$ can read all $C_j,j\ne i$ on blockchain,
                  $\mathbf{Sim}_i$ computes $C_j$ by encrypting zeros which is computationally indistinguishable since $C_j$ is ciphertext encrypted by semantically secure $PKE_1$.
                  $\mathbf{Sim}_i$ can compute $(D,A_2)_{j,i}$ for all $j\ne i$ with $y_i$.
                  This is the same as the real.
                  In step 3, $\mathbf{B_i}$ can read all $(D,A_2)_{i,j}$ for $j\ne i$ on blockchain. Since $D$ is encrypted by $PKE_1$, $\mathbf{Sim}_i$ computes $D$ by encrypting zeros which is indistinguishable.
                  In the real, $\mathbf{B_i}$ can decrypts $A_2$ by the private key of $PKE_2$ obtaining the messages encrypted by $PKE_3$ which is again semantically secure. $\mathbf{Sim}_i$ computes $A_2$ by encrypting random numbers selected from
                  $\mathcal{M}_{\oplus}$ via $PKE_2$ and $PKE_3$, which is indistinguishable.
                  $\mathbf{Sim}_i$ decrypts all $D$ getting $w$, and computes $(A_1)_{j,i}$ by the blind and shuffle operation.
                  Just as $A_2$, $\mathbf{Sim}_i$ can compute $A_1$ that is computationally indistinguishable. Note that $A_1$, as well as $A_0$, will indicate whether $\mathbf{B_i}$ is the winner or not, it should be simulated according to $\mathcal{F}_i$. If $\mathbf{B_i}$ is a winner, $\mathbf{Sim}_i$ computes $(A_{1,\ell})_{i,j}$ by encrypting uniformly and randomly chosen non-zero numbers for every $\ell$ and $j,j\ne i$, then $\mathbf{Sim}_i$ obtains $(A_{1})_{i,j}=((A_{1,k-1})_{i,j},\dots,(A_{1,0})_{i,j})$. Otherwise, for each $j,j\ne i$, $\mathbf{Sim}_i$ randomly chooses $\zeta\in\{k-1,\dots,0\}$, gets $(A_{1,\zeta})_{i,j}$ by encrypting 0, and computes $(A_{1,\ell})_{i,j},\ell\ne\zeta$ by encrypting uniformly and randomly chosen non-zero numbers.
                  In step 4, $\mathbf{Sim}_i$ decrypts $A_1$ and gets $A_0$ after blinding and shuffling.

                  As for $\mathbf{A}$, we construct $\mathbf{Sim}_A$ given related public and private keys and the output $\mathcal{F}_A$ of $\mathbf{A}$ receiving from TTP. As discussed above, in the view of $\mathbf{A}$, all messages communicated are ciphertexts encrypted by (semantically) secure PKE, thus they are simulatable by encrypting mud numbers, which are computationally indistinguishable. Since $\mathbf{A}$ will learn the comparison results from $(A_0)_{i,j}$, it should be simulated according to $cmp_{i,j}$. If $y_i\ge y_j$, then for all $\ell$, $(A_{0,\ell})_{i,j}$ should not be decrypted to 0. In this case, $\mathbf{Sim}_A$ encrypts $\ell$ non-zero random numbers by $PKE_3$ then shuffles them. Otherwise, there exists $\zeta$ such that $\textbf{Dec}_{\oplus}((A_{0,\zeta})_{i,j})=0$. $\mathbf{Sim}_A$ encrypts 0 and $(\ell\!-\!1)$ non-zero random numbers then shuffles them. For the winner, $\mathbf{Sim}_A$ can perform similarly to the case (step 3) where above $\mathbf{Sim}_i$ simulates that $\mathbf{B_i}$ is a winner.

                  For a complaint from $\mathbf{B}$ about $\mathbf{B'}$, in the real, $\mathbf{A}$ runs TPIC with $\mathbf{B}$ and $\mathbf{B'}$. The proof of this case is the same as Theorem~1.
                  So we have
                  $$\mathbf{Sim}_A(\mathcal{F}(In_\mathcal P),Out_A,\phi_A)\approx_c \mathbf{VIEW}^{\Pi}_{A}(In_\mathcal P,\mathcal{F}(In_\mathcal P),Out_A,\phi_A),$$
                  where $\mathcal{P}=\{\mathbf{B_i}:i=1,\dots,N\}$ and $\Pi$ denotes our scheme.
                  Hence, our scheme is secure against adversary $\mathcal{A}_1$ who may control the auctioneer.

                  For $\mathcal{A}_2$ who may control more than one bidders denoted by set $I$, we construct $\mathbf{Sim}_I$ as follows. Since the bidders are insulated in the real, they computes ciphertexts independently. $\mathbf{Sim}_I$ does what $\mathbf{Sim}_j$ does for all $j\in I$, such that
                  $$\mathbf{Sim}_I(In_I,\mathcal{F}(In_\mathcal P),\phi_I)\approx_c \mathbf{VIEW}^{\Pi}_{I}(In_\mathcal P,\mathcal{F}(In_\mathcal P),\phi_I).$$
                  Thus the scheme is also secure against adversary $\mathcal{A}_2$.
                 \end{proof}

         The complexity analysis and comparisons with related work are given in a followed section.

\section{Complexity Analysis of the Auction}\label{sect-complexity}
            In this section, we analyse the performance of our auction scheme. It is assumed that the number of bidders is $N$ and the bids are $\eta$ bits, i.e., $k=\lceil\log_d(2^\eta)\rceil,d>2$.

             Our auction scheme runs in 4 rounds (blocks) which is constant in both the length of the bid and the number of bidders. In the first round, each bidder just encrypts his bid to $k$ ciphertexts. In other rounds, a bidder computes $k(N\!-\!1)$ different ciphertexts comparing his bid with other $N\!-\!1$ bids in parallel, then he publishes them on blockchain.
             While $PKE_3$ is more efficient over a fast elliptic curve, $PKE_1$ and $PKE_2$ are much time consuming performed over finite field.
             Referring to the experimental data in \cite{CEK}, we roughly estimate the time cost as shown in Table~\ref{table_time}. The time is computed according to the time to compute one ciphertext since each bidder computes $k(N\!-\!1)$ different ciphertexts. To meet the cryptographic parameters in \cite{CEK},
             we set $b=2,d=8$, $N=100$ for an auction with 100 bidders, and the bid with bit-length 30($|y_i|=30$, a ``billion" magnitude), thus $k=10$. For 128-bit security level, set $|n|=3072,|u|=256$, and $|n|=15360,|u|=512$ for 256-bit security. For $PKE_3$, the exponential variant of ElGamal \cite{ST} is implemented over an elliptic curve. We refer the reader to the original paper \cite{CEK} for more details about specific parameterizations.

            \begin{table}[!t]
            \begin{footnotesize}
      \caption{Performance analysis of our auction scheme.}
      \label{table_time}
         \begin{center}
                \begin{tabular}{ccccccc}
                    \toprule
                    \     \quad &\       \quad &\ 128-bit&security     \quad  &\ 256-bit&security \\
                    Round \quad &Com\quad&Time(s)  \quad &Size(MB)\quad &Time(s)\quad &Size(MB)\\
                    \midrule
                    \#1   \quad &$0$       \quad&0.39     \quad &0.37   \quad &8.52   \quad &1.83\\
                    \#2   \quad &$k(N\!-\!1)$    \quad&1.14     \quad &0.73   \quad &15.66   \quad &3.66 \\
                    \#3   \quad &$k(N\!-\!1)$    \quad&0.81     \quad &0.37   \quad &9.21    \quad &1.83 \\
                    \#4   \quad &$k(N\!-\!1)$    \quad&0.14     \quad &0.12   \quad &2.13    \quad &0.25 \\
                \bottomrule
                \end{tabular}
         \end{center}
      \end{footnotesize}
      \end{table}

             The performance analysis of our auction scheme is given in Table~\ref{table_time}. It shows the computation times of $PKE_2$(\emph{Com}), encryption and decryption time of $PKE_1$ and $PKE_3$(\emph{Time}), the size of transferred messages(\emph{Size}). In the table, we just list the computation times of $PKE_2$, to make it easier to be compared with the related comparison protocols, and fix $PKE_2$ to be RSA. First, our scheme is practical for real-world application.
             In step 2,3, each bidder $\mathbf{B_i}$ needs to encrypt $k(N\!-\!1)$ numbers by RSA($PKE_2$), but, in step 4, each $\mathbf{B_i}$ is supposed to decrypt $k(N\!-\!1)$ ciphertexts.
             For 128-bit security ($|n|=3072$), we test the RSA running time in Go language on PC (AMD PRO A10-8770 R7, 10 COMPUTE CORES 4C+6G) running at 3.5GHz, which shows that the RSA encryption time is about 230ms for $k=10,N=100$ while about 9s is consumed to perform decryptions.
             Thus, the total time in each round is about 10 seconds, which is acceptable in reality even in permissionless blockchain architectures like Bitcoin and Ethereum requiring 10 minutes and 15s (on average) to produce a new block.

             To illustrate the efficiency, we compare our scheme with the two closely related works, DGK \cite{DGK} and Fischlin \cite{Fisch} protocols, which are secure against semi-honest adversary. These three schemes all use an RSA modulus for the encryption. It is certainly reasonable to use the same bit length of the modulus. We say our scheme is a bit faster than (or, at least, comparable to) the other two comparison protocols as claimed in \cite{CEK} that their scheme (based on which we present our TPIC protocol) is faster about 3.5 times than the DKG protocol, which is roughly 10 times \cite{DGK} faster than Fischlin protocol.
             To compare $\eta$-bit numbers, our scheme needs to compute $k\!=\!\lceil\log_d(2^\eta)\rceil$ ciphertexts while the other two protocols require $\eta$ ciphertexts since they are per-bit comparisons. Moreover, to achieve $AND$ operation in Fischlin protocol, each ciphertext is expanded to a tuple with $\iota$ elements where $\iota$ is the soundness error, which adds the communication overhead.
             For round complexity, our scheme runs in 4 rounds, which is the same as DGK, while Fischlin protocol requires 6 rounds \cite{BK2} to get the comparison results.
             We also note that $S\!train$ \cite{BK2} even with all the NIZK related messages removed still requires a larger communication size. The performance comparison is shown in Table~\ref{table_compare} including round complexity (\emph{round}), the number of ciphertexts needed to compare two integers (\emph{C-num}) and communication overhead (\emph{com(bits)}).
             Besides, an instance is given in the \emph{com-eg} row by setting $\eta=30,b=2,d=8,k=10,N=100,|n|=3072,\iota=40$ to show the differences.
      \doublerulesep 0.1pt
      \begin{table}[!t]
      \begin{footnotesize}
      \caption{Performance comparisons with related schemes.}
      \label{table_compare}
         \begin{center}
                \begin{tabular}{ccccc}
                    \toprule
                    \ --   \quad & Fischlin            \quad &DGK     \quad &$Strain$ \quad &Our scheme \\
                    \midrule
                    round   \quad &6                  \quad&4        \quad&4         \quad &4   \\
                    C-num \quad &$\eta \iota  $ \quad&$\eta    $ \quad&$\eta\iota$ \quad &$k$ \\
                    com(bits)\quad &$\eta\iota|n|N$ \quad&$\eta|n|N$ \quad&$\eta\iota|n|N$ \quad &$k|n|N$ \\
                    com-eg(KB)\quad &45000               \quad&1125        \quad&45000      \quad &375   \\
                \bottomrule
                \end{tabular}
         \end{center}
      \end{footnotesize}
      \end{table}

      Here, we note that $S\!train$ is not suitable for a large auction scheme since it may need to create a transaction with heavy payload (not counting NIZK messages, 4MB for 10 bidders and 40MB for 100 bidders) while a light transaction less than 400KB is required in our scheme. An auction can be done in 5 blocks (an additional block to broadcast the winner) in several minutes considering there may be complaints submitted. If there is a complaint, the auctioneer $\mathbf{A}$ needs to handle it and publishes complaint result in another block wasting a block of time.
      At the beginning of the auction, the auctioneer $\mathbf{A}$ can set a deadline (e.g. several blocks in the future) for bidders to submit commitments and deposits similar to registration in the real.
      Then, each bidder $\mathbf{B_i}$ submits a commitment to his bid via a transaction transferring some deposits to $\mathbf{A}$. The deposits can be redeemed by $\mathbf{B_i}$ revealing his bid or refunded by the auctioneer at the end of the auction. Since the auctioneer is assumed to be semi-honest, the honest bidders will not lose their deposits.

\section{Conclusions}\label{sect-conclu}
            In this paper, three-party comparison protocols are given, which are executed among a judge and two competitors. The judge learns the comparison results and the ranking while protecting the private integers of competitors. As applications of three-party comparisons, multi-party integer comparisons including joint bidding and online auction, are presented. All schemes run in a
            constant number of blocks with the help of blockchain. Besides, blockchain-based commitment is used to encourage the competitors to perform correctly.
           Security proofs of the comparisons  are also given. The complexity analysis of the auction shows that our scheme performs well with lower communication overhead and comparable time cost, compared with related designs.

\bibliographystyle{fundam}
\bibliography{00MPIC}


\end{document}